\documentclass[journal,twoside,web]{ieeecolor}
\usepackage{generic}
\usepackage{lcsys}
\usepackage{mypreamble}

\newcommand{\maketextstyle}{\textstyle} %

\newcommand{\PP}{PP}
\renewcommand{\PP}{{period-preserving}}

\newcommand{\blue}[1]{\textcolor{blue}{#1}}
\newcommand{\mathblue}[1]{\mathcolor{blue}{#1}}

\renewcommand{\blue}[1]{#1}
\renewcommand{\mathblue}[1]{#1}

\usepackage[normalem]{ulem}

\begin{document}

\title{Amplitude-Dependent Bode Diagrams via Scaled Relative Graphs}
\author{Julius P.~J. Krebbekx$^1$, Roland Tóth$^{1,2}$, Amritam Das$^1$, Thomas Chaffey$^3$
\thanks{$^1$Control Systems group, Department of Electrical Engineering, Eindhoven University of Technology, 5600 MB Eindhoven, The Netherlands.}
\thanks{$^2$Systems and Control Lab, HUN-REN Institute for Computer Science and Control, 1111 Budapest, Hungary.}
\thanks{$^3$School of Electrical and Computer Engineering, The University of Sydney, Australia. 
E-mails: {\tt\small \{j.p.j.krebbekx, r.toth, am.das\}@tue.nl, thomas.chaffey@sydney.edu.au}}
\rlap{\smash{
\begin{tikzpicture}[remember picture,overlay]
  \node[anchor=north, yshift=-0.2cm] at (current page.north) {
    \parbox{0.9\textwidth}{
      \scriptsize \centering
      \copyright~2026 IEEE. Personal use of this material is permitted.  Permission from IEEE must be obtained for all other uses, in any current or future media, including reprinting/republishing this material for advertising or promotional purposes, creating new collective works, for resale or redistribution to servers or lists, or reuse of any copyrighted component of this work in other works. This work has been accepted for publication in IEEE Control Systems Letters. This is the author's accepted version prior to copyediting, proofreading, and typesetting.
    }
  };
\end{tikzpicture}}}%
}

\maketitle
\thispagestyle{empty}

\begin{abstract}
    \emph{Scaled Relative Graphs} (SRGs) provide an intuitive graphical frequency-domain method for the analysis of \emph{Nonlinear} (NL) systems, generalizing the Nyquist diagram. In this paper, we develop a method for computing $L_2$-gain bounds for Lur'e systems over bounded frequency and amplitude ranges. We do this by restricting the input space of the SRG both in frequency and energy content, and using methods from Sobolev theory. The resulting gain bounds over restricted sets of inputs are less conservative than bounds computed over the entire $L_2$, and yield three-dimensional NL generalization of the Bode diagram, plotting $L_2$-gain as function of both input frequency and energy content. In the zero-energy limit, the \emph{Linear Time-Invariant} (LTI) Bode diagram is recovered, and at the infinite-energy zero-frequency limit, we recover the $L_2$-gain. The effectiveness of our method is demonstrated on an example that resembles Phase-Locked Loop dynamics.
\end{abstract}


\section{Introduction}

\IEEEPARstart{F}{or} \emph{Linear Time-Invariant} (LTI) systems, graphical analysis tools such as the Nyquist and the Bode diagram are foundational for control engineering. They enable intuitive and easy to use controller design methods through their frequency-domain interpretation. When pushing the performance of dynamical systems, the \emph{Nonlinear} (NL) dynamics start to play a crucial role. Inspired by the success of analyzing performance through Bode diagrams in the LTI case, e.g., for mixed sensitivity shaping~\cite{skogestadMultivariableFeedbackControl2010}, there have been multiple attempts to generalize the Bode diagram to the NL case. The describing function~\cite{krylovIntroductionNonlinearMechanics1947}, which is approximate, has been the first step in this direction by considering both frequency- and amplitude-dependence, and has been successfully applied in practice~\cite{heertjesVariableGainMotion2016}. Most modern and non-approximate methods focus on sinusoidal inputs~\cite{pavlovFrequencyDomainPerformance2007}, but extensions to higher-order harmonics in the input~\cite{rijlaarsdamComparativeOverviewFrequency2017} exist as well. %

The \emph{Scaled Relative Graph} (SRG) \cite{ryuScaledRelativeGraphs2022} is a novel graphical analysis method and was first used in~\cite{chaffeyGraphicalNonlinearSystem2023} to generalize the Nyquist diagram to a NL setting. It is an exact method, and it is intuitive because of its close connection to the Nyquist diagram. At its core, SRG analysis provides performance bounds in terms of (incremental) $L_2$-gain. However, many researchers have realized that $L_2$-gain is too conservative to serve as a meaningful performance metric for NL systems, and instead, the set of inputs should be restricted~\cite{megretskiCombiningL1L21995,megretskiSystemAnalysisIntegral1997,dahlehRejectionPersistentBounded1992,zamesNonlinearOperatorsSystem1960}. 

\blue{In~\cite{krebbekxNonlinearBandwidthBode2025a}}, the SRG has been used to compute the $L_2$-gain restricted to spaces of periodic inputs, yielding NL Bode diagrams~\cite{krebbekxNonlinearBandwidthBode2025a}. A major limitation of~\cite{krebbekxNonlinearBandwidthBode2025a} is that amplitude is not considered; hence, the gains are computed as worst-case gains over all amplitudes of inputs with a certain frequency.

In this paper, we extend~\cite{krebbekxNonlinearBandwidthBode2025a} by considering an additional energy constraint on the input, which yields an amplitude bound on the output. Our approach consists of two parts. First, we develop a method for bounding the output amplitude for inputs of bounded frequency and amplitude using Sobolev theory~\cite{brezisFunctionalAnalysisSobolev2011,ziemerWeaklyDifferentiableFunctions1989}. This method relies on the $L_2$-gain from the input to the output and from the time derivative of the input to the time derivative of the output. In the second part, SRG methods are developed to compute these $L_2$-gains from the input to the output and their time derivatives. %
The result is a three-dimensional Bode diagram; $L_2$-gain as a function of both input frequency and energy. We demonstrate our method on a Lur'e system that represents \emph{Phase-Locked Loop} (PLL) dynamics~\cite{stephensPhaseLockedLoopsWireless1998}.

\blue{Recently, a notion of nonlinear Bode diagrams with amplitude dependence was developed in~\cite{moreschiniFrequencyResponseLoop2026}. However, their approach differs significantly from the one presented here in several ways. In~\cite{moreschiniFrequencyResponseLoop2026}, an exact amplitude-dependent gain is computed using the (numerical) solution of the differential equations describing the system via the invariance principle, focusing on general nonlinear systems and parametrized families of input signals. In contrast, the method presented in this paper considers all periodic signals with finite energy, not just a parametrized set, and does not require solving the underlying differential equations of the system. On the other hand, we restrict our focus to Lur'e systems and compute an upper bound for the amplitude-dependent gain using SRGs, instead of the exact gain.}

The rest of this paper is structured as follows. In Section~\ref{sec:preliminaries}, we present the required preliminaries. Our first contribution, the frequency-dependent gains and the amplitude bound from Sobolev theory, is developed in Section~\ref{sec:ampl-freq-dep-gain}. Our second contribution is given in Section~\ref{sec:gain-from-SRGs}, where SRG tools are developed to compute frequency- and amplitude-dependent gains. Finally, we apply our results to the PLL example in Section~\ref{sec:example} and present our conclusions in Section~\ref{sec:conclusion}.

\section{Preliminaries}\label{sec:preliminaries}

Let $\R$ and $\C$ ($\C_\infty := \C \cup \{ \infty \}$) denote the real and (extended) complex numbers, respectively, with $\R_{\geq 0} = [0, \infty)$. For sets $A, B \subseteq \C_\infty$, the sum $A+B$ and product $AB$ are the Minkowski sum and product. We denote the disk in $\C$ centered on $\R$, which intersects $\R$ in $[\alpha, \beta]$, as $D_{[\alpha, \beta]}$. The distance between \blue{non-empty} sets $\mathcal{C}_1,\mathcal{C}_2 \subseteq \C_\infty$ is defined as $\dist(\mathcal{C}_1,\mathcal{C}_2) := \inf_{z_1 \in \mathcal{C}_1, z_2 \in \mathcal{C}_2} |z_1-z_2|$, where $|\infty-\infty|:=\blue{\infty}$.

\subsection{Signals}

We consider the sets of signals $\mathcal{V} = \{u : \R_{\geq 0} \to \R\}$ and $\mathcal{V}_T = \{u : [0,T] \to \R\}$, and the signal spaces $L_2 = \{ u \in \mathcal{V} \mid \norm{u} < \infty \}$, where $\norm{u}^2 = \int_0^\infty |u(\tau)|^2 \dd \tau$, and $L_2[0,T] = \{ u \in \mathcal{V}_T \mid \norm{u}_T < \infty \}$, where $\norm{u}_T^2 = \int_0^T |u(\tau)|^2 \dd \tau$ for any $T>0$. The supremum norm is denoted as $\norm{u}_\infty$. Define the truncation operator $P_T$ on signals $u \in \mathcal{V}$ as $(P_T u)(t) = 0$ for all $t>T$, else $(P_T u)(t) = u(t)$. The extension of $L_2$ %
is defined as $\Lte := \{ u \in \mathcal{V} \mid \norm{P_T u} < \infty \text{ for all } T \in \R_{\geq 0} \}$.
The space $\Lte$ is particularly useful since it includes periodic, persistent and unbounded signals, which are otherwise excluded from $L_2$.

Periodic signals $u\in \Lte$ can also be uniquely viewed as elements in $L_2[0,T]$, where $T$ is the period of the signal, i.e. $u(t)=u(t+T)$ for all $t\in [0,\infty)$. Any $u \in L_2[0,T]$ can be written in Fourier series as $u(t) = \sum_{k \in \Z} \hat{u}_k e^{2 \pi j k t/T}$ where $\hat{u}_k \in \C$ are the Fourier coefficients. The \emph{Root-Mean-Square} (RMS) norm of the signal is defined as 
\begin{equation}\label{eq:rms_fourier_L2}
    \maketextstyle \norm{u}_\mathrm{RMS} := \sqrt{\sum_{k \in \Z} |\hat{u}_k|^2} = 1/\sqrt{T} \norm{u}_T.
\end{equation}
For any locally integrable\footnote{$\int_I |u(\tau)| \dd \tau <\infty$ for all compact $I \subseteq \R_{\geq 0}$.} $u$, the \emph{weak derivative} is a function $v$ that obeys $\int_{\R_{\geq 0}} u(\tau) \phi'(\tau) \dd \tau = - \int_{\R_{\geq 0}} v(\tau) \phi(\tau) \dd \tau$ for all continuously differentiable $\phi$ with local support~\cite[Ch. 8]{brezisFunctionalAnalysisSobolev2011}. If such $v$ exists, we denote it with $u'$. If $u$ has a locally integrable derivative $\frac{\dd}{\dd t}u$, then $\frac{\dd}{\dd t}u = u'$ holds. 

Let $H^1 = \{ u \in \mathcal{V} \mid u, u' \in L_2\}$ denote the \emph{Sobolev space} of signals for which both $u$ and $u'$ have finite energy. Similarly, let $H^1[0,T] = \{ u \in \mathcal{V}_T \mid u,u' \in L_2[0,T] \} $ and define the spaces $H^1_{t_0} = \{u \in H^1 \mid u(t_0) = 0\}$ and $H^1_{t_0}[0,T] = \{u \in H^1[0,T] \mid u(t_0) = 0\}$. %
The latter two spaces are of interest as they enforce $u(t_0) = 0$.

\subsection{Systems}\label{sec:systems}

For operators $R :L_2 \to L_2$, we define the $L_2$-gain as
\begin{equation}\label{eq:L2-gain}
    \maketextstyle \gamma(R) := \sup_{\mathblue{0 \neq} u \in L_2} \frac{\norm{Ru}}{\norm{u}}.
\end{equation}
We model systems as causal\footnote{An operator $R : \Lte \to \Lte$ is \emph{causal} if $P_T R P_T = P_T R$.} operators $R : \Lte \to \Lte$ with zero initial conditions (if applicable). Moreover, we assume $R(0) =0$. A system $R$ is called $L_2$-stable if $Ru \in L_2$ for all $u \in L_2$, and it has finite $L_2$-gain if $\gamma(R) < \infty$.

If $R : H^1 \to H^1$, then the map $\partial R$, defined as $ u' \mapsto y'$ where $y = Ru$, is called the \emph{derivative map}. 

The feedback system in Fig.~\ref{fig:lure}, where $G: \Lte \to \Lte$ is LTI and $\phi :\R \to \R$ is a NL function, defined by
\begin{equation}\label{eq:lure}
    y = G e, \quad e = u - \phi(y),
\end{equation}
is called a \emph{Lur'e system}. Such a system is called \emph{well-posed}~\cite{megretskiSystemAnalysisIntegral1997} if the map $e \mapsto u$, defined by $u = e+\phi(Ge)$ in~\eqref{eq:lure}, has a causal inverse on $\Lte$. We will use the notation $[G, \phi]$ to denote the map $u \mapsto y$ defined by~\eqref{eq:lure}\blue{.} We will often require the following condition on the Lur'e system.%

\begin{condition}\label{condition:lure}
    In~\eqref{eq:lure}, $G$ is a proper and stable\footnote{All poles $p$ obey $\mathrm{Re}(p) < 0$.} LTI system, $\gamma(\phi) < \infty$, and $[G, \tau \phi]$ is well-posed for all $\tau \in [0,1]$.
\end{condition}

\begin{figure}[t]
    \centering

    \tikzstyle{block} = [draw, rectangle, 
    minimum height=1.5em, minimum width=2em]
    \tikzstyle{sum} = [draw, circle, scale=0.5, node distance={0.5cm and 0.5cm}]
    \tikzstyle{input} = [coordinate]
    \tikzstyle{output} = [coordinate]
    \tikzstyle{pinstyle} = [pin edge={to-,thin,black}]
    
    \begin{tikzpicture}[auto, node distance = {0.1cm and 0.5cm}]
        \node [input, name=input] {};
        \node [sum, right = of input] (sum) {};
        \node [block, right = of sum] (lti) {$G(s)$};
        \node [coordinate, right = of lti] (z_intersection) {};
        \node [output, right = of z_intersection] (output) {}; %
        \node [block, below = of lti] (static_nl) {$\phi$};
    
        \draw [->] (input) -- node {$u$} (sum);
        \draw [->] (sum) -- node {$e$} (lti);
        \draw [->] (lti) -- node [name=z] {$y$} (output);
        \draw [->] (z) |- (static_nl);
        \draw [->] (static_nl) -| node[pos=0.99] {$-$} (sum);
    \end{tikzpicture}
    
    \caption{Lur'e system with LTI $G$ and NL $\phi : \R \to \R$.}
    \label{fig:lure}
    \vspace{-1.5em}
\end{figure}

\subsection{Scaled Relative Graphs}

\blue{In this paper, we will derive a bound on the output amplitude via an energy bound on the input. As} amplitude is measured with respect to a reference zero value, we use the non-incremental version of the SRG, called the \emph{Scaled Graph} (SG). For linear operators, the SRG and SG coincide~\cite{chaffeyGraphicalNonlinearSystem2023}.

Let $\mathcal{L}$ be a Hilbert space with inner product $\inner{\cdot}{\cdot}: \mathcal{L} \times \mathcal{L} \to \C$ and consider the relation $R : \mathcal{L} \to \mathcal{L}$. The angle between $u, y\in \mathcal{L}$ is defined as $\angle(u, y) := \cos^{-1} \frac{\mathrm{Re} \inner{u}{y}}{\norm{u} \norm{y}} \in [0, \pi]$\blue{, where $\angle(u, y) := 0$ if either $u$ or $y$ is zero.} For any $\mathcal{U} \subseteq \mathcal{L}$, the SG of $R$ over $\mathcal{U}$ is defined as
\begin{equation}\label{eq:SG}
    \textstyle \SG_\mathcal{U}(R) = \left\{ \frac{\norm{y}}{\norm{u}} e^{\pm j \angle(u, y)} \mid y \in Ru, \, \mathblue{0 \neq} u \in \mathcal{U} \right\}.
\end{equation}
If $\mathcal{U} = \mathcal{L}$, we write $\SG(R)$. It follows from~\eqref{eq:L2-gain} and~\eqref{eq:SG} that $\gamma(R)$ is the radius of $\SG(R)$.

We use the following stability condition\footnote{For non-incremental analysis using the SG, well-posedness must be assumed. However, when using SRGs for incremental analysis, well-posedness is \emph{constructed} from causal loop components and the incremental small-gain theorem; see~\cite{chaffeyHomotopyTheoremIncremental2025}. Moreover, the homotopy condition $\tau \in [0,1]$ can be omitted when using the hard~\cite{chenSoftHardScaled2026} or extended~\cite{krebbekxScaledRelativeGraph2025a} SRG.} for~\eqref{eq:lure} from~\cite{vandeneijndenScaledGraphsReset2024}, where the gain bound is extracted at $\tau=1$ as in~\cite{chaffeyHomotopyTheoremIncremental2025}.

\begin{proposition}\label{prop:lure_stability_theorem}
    Suppose that Condition~\ref{condition:lure} holds, $\SG(\phi)$ satisfies the chord property (see~\cite{ryuScaledRelativeGraphs2022} for the definition), then $\gamma([G,\phi]) \leq 1/r_1$ holds if there exist $r, r_\tau $ such that
    \begin{equation}\label{eq:lure_stability_theorem}
        \dist(\SG(G)^{-1}, -\tau \SG(\phi)) \geq r_\tau \geq r >0, \; \forall \tau \in [0,1].
    \end{equation}
\end{proposition}

In this paper, we bound $\SG(\phi) \subseteq D_{[a,b]}$ with some $a\leq b \in \R$. Note that any such disk $D_{[a,b]}$ satisfies the chord property.

\section{Amplitude and Frequency Dependent Gains}\label{sec:ampl-freq-dep-gain}

In this section, we define the frequency- and amplitude-dependent gains that are the subject of this paper. We also show how Sobolev theory can be used to bound $\norm{y}_\infty$ in~\eqref{eq:lure}.

\subsection{Frequency-Dependent Gain}

Note that $L_2$-gain~\eqref{eq:L2-gain} alone provides no frequency-domain information. To obtain a more insightful performance metric, we consider the subspace of odd $T$-periodic signals\footnote{\blue{We use $ \mathcal{U}_{\omega \mathrm{e}} \subseteq \Lte$ to refer to the periodic signals in $\Lte$ obtained by repeating $u \in \mathcal{U}_\omega$ on $[k T, (k+1)T]$ for $k=0,1,\dots$.}} 
\begin{equation}\label{eq:input_U_w}
    \maketextstyle \mathcal{U}_\omega = \left\{ u \in L_2[0,T] \mid u(t) = \sum_{k \in 2\Z+1} \hat{u}_k e^{j \omega k t} \right\},
\end{equation}
with frequency $\omega = \frac{2 \pi}{T}$ and symmetry $u(t+T/2)=-u(t)$.

\begin{definition}
    A signal $u \in \Lte$ is called odd $T$-periodic for some $T>0$ if $u(t+T/2) = -u(t)$ for all $t \geq 0$.
\end{definition}

\begin{definition}
    The Lur'e system in Fig.~\ref{fig:lure} is called period-preserving if for all odd $T$-periodic $u \in \Lte$, the signals $e, y$ converge asymptotically\footnote{I.e., $\forall \epsilon>0 , \exists t>0 : \forall \tau \geq t, \; |y(\tau)-\tilde{y}(\tau)| < \epsilon$ (idem for $e \to \tilde{e}$).} to odd $T$-periodic $\tilde{e},\tilde{y} \in \Lte$.
\end{definition}

If the output $y = Ru$ for an odd $T$-periodic input $u \in \mathblue{\mathcal{U}_{\omega \mathrm{e}}}$ converges asymptotically to some $\tilde{y} \in \mathblue{\mathcal{U}_{\omega \mathrm{e}}}$, then we can define the frequency-dependent gain, as in~\cite{krebbekxNonlinearBandwidthBode2025a}, as
\begin{equation}\label{eq:gamma_w}
    \maketextstyle \gamma_\omega(R) := \smash{\sup\limits_{\mathblue{0 \neq } u \in \mathblue{\mathcal{U}_{\omega \mathrm{e}}}} \lim\limits_{t \to \infty }} \frac{\norm{P_t y}}{\norm{P_t u}} =  \smash{ \sup\limits_{\mathblue{0 \neq }u \in \mathcal{U}_\omega} }\frac{\norm{\tilde{y}}_\mathrm{RMS}}{\norm{u}_\mathrm{RMS}} , 
\end{equation}
where the last equality follows from~\cite[Lemma 1]{krebbekxNonlinearBandwidthBode2025a}.

For odd $T$-periodic signals, our approach restricts the \blue{SG} of the LTI block $G$ to the frequencies present in the steady-state periodic response $\tilde{e},\tilde{y}$. Nonlinearities under period preservation excite higher harmonics, making the framework most effective for characterizing high-frequency roll-off. A nonzero DC component would impose a high-frequency gain lower bound tied to the DC gain of $G$. To guarantee period preservation without DC components, we exploit odd half-wave symmetry: odd nonlinearities induce odd $T$-periodic solutions with zero mean (see Theorem~\ref{thm:lure}). %
One can include even harmonics, but ensuring periodicity is harder (e.g., period preservation is an assumption in~\cite{krebbekxNonlinearBandwidthBode2025a}).

\subsection{Amplitude-Dependent Gain}

Drawing inspiration from Sobolev theory, we define a quantity called the \emph{harmonic energy} $U = \norm{u}_T \norm{u'}_T$, which captures both the energy and harmonic content of a periodic signal $u \in L_2[0,T]$ with period $T$. By bounding the harmonic energy of the input, we will be able to obtain bounds not only on $\norm{y}$, but also on $\norm{y}_\infty$. %

To motivate the choice of $U$, consider an input $u(t) = A \sin (\omega t)$, then, $\norm{u}_T \norm{u'}_T  = A^2 \pi$, i.e., $U$ is independent of $T$. Equivalently, via~\eqref{eq:rms_fourier_L2}, the quantity $U$ equals the RMS power applied for a time $T$. When $u(t) = A \sin(\omega t) + B \sin(3 \omega t)$, then $\norm{u}_T \norm{u'}_T = \pi \sqrt{ (A^2 +B^2)(A^2+ 9B^2)}$. Clearly, $U$ captures both the energy in $\norm{u}_T$ and harmonic content through $\norm{u'}_T$. Since we are interested in the output amplitude $\norm{y}_\infty$, it is natural to consider not only the energy of the input, but also its rate of variation. Therefore, $U$ is an adequate parameter to measure energy and harmonic content on both $L_2$ and $L_2[0,T]$ for all $T>0$. %

By restricting the harmonic energy of the input, we obtain a frequency- and amplitude-dependent gain, defined as
\begin{equation}\label{eq:gamma_w_U}
    \maketextstyle \gamma_{\omega, U}(R):= \underset{\mathblue{0 \neq } u \in \mathcal{U}_\omega, \norm{u}_T \norm{u'}_T \leq U}{\sup} \frac{\norm{\tilde{y}}_\mathrm{RMS}}{\norm{u}_\mathrm{RMS}}, 
\end{equation}
and its $L_2$-version as $ \gamma_{U}(R):= \sup_{\mathblue{0 \neq } u \in L_2, \norm{u} \norm{u'} \leq U} \frac{\norm{y}}{\norm{u}}$.

\subsection{Bounding Amplitude using Sobolev Methods}

In this section, we show that by bounding the harmonic energy $U$ of the input, we obtain a bound on $\norm{y}_\infty$.

\begin{lemma}\label{lemma:gagliardo-nirenberg}
    Let $u \in H^1_{t_0}[0,T]$ and $\norm{u}_T \norm{u'}_T \leq U$, then
    \begin{equation}\label{eq:sobolev}
        \norm{u}_\infty \leq \sqrt{2 \norm{u}_T \norm{u'}_T } \leq \sqrt{2 U}.
    \end{equation}
\end{lemma}
\begin{proof}
    For $u \in H^1$ and $u \in H^1[a,b]$, the chain rule $(u^2)'=2 u u'$ and the fundamental theorem of calculus hold for weak derivatives~\cite[Ch. 8]{brezisFunctionalAnalysisSobolev2011}, \blue{hence} $u(t)^2 - u(t_0)^2 = \int_{t_0}^t [u(\tau)^2]' \dd \tau = 2 \int_{t_0}^t u(\tau) u'(\tau) \dd \tau$.
    By taking the absolute value and applying Cauchy-Schwartz in the previous equation, we obtain $|u(t)^2-u(t_0)^2| = \mathblue{|2 \int_{t_0}^t u(\tau) u'(\tau) \dd \tau | \leq}  2 \int_{\min\{t_0,t\}}^{\max\{t_0,t\}} |u(\tau)| |u'(\tau)| \dd \tau \leq 2 \int_0^T |u(\tau)| |u'(\tau)| \dd \tau \leq 2\norm{u}_T \norm{u'}_T$.
    The bound\footnote{For non-periodic inputs, we restrict to $H^1_0$, i.e., inputs that start at zero. For $u \in H^1_0$, the same derivation yields $\norm{u}_\infty \leq \sqrt{2 \norm{u} \norm{u'} }$.}~\eqref{eq:sobolev} follows \blue{by choosing $u(t_0)=0$,} and by taking the supremum over $t$ and the square root.
\end{proof}
The bound in~\eqref{eq:sobolev} is a special case of the Gagliardo-Nirenberg interpolation inequality~\cite[p. 233]{brezisFunctionalAnalysisSobolev2011}. 

Note that $u(t_0) =0$ at some $t_0 \in [0,T]$ holds for all $u \in \mathcal{U}_\omega \cap H^1[0,T]$, since $u \in H^1[0,T]$ are absolutely continuous and $u(t+T/2) = -u(t)$ for $u \in \mathcal{U}_\omega$. %

\subsection{Computing Frequency- and Amplitude-Dependent Gains}

In order to use Lemma~\ref{lemma:gagliardo-nirenberg} to bound $\norm{y}_\infty$ for the Lur'e system, we must restrict the input to $u \in \mathcal{U}_\omega \cap H^1[0,T]$ (or $H^1_0$). Assuming~\eqref{eq:lure} is \PP{}, let the associated periodic steady-states be $\tilde{e}, \tilde{y}$. For $\norm{\tilde{y}}_\infty \leq A$, we define the gain bounds $\norm{\tilde{y}}_T \leq \lambda_{\omega, A} \norm{u}_T$ and $\norm{\tilde{y}'}_T \leq \lambda^\partial_{\omega, A} \norm{u'}_T$, which can be computed using SRG methods as we show in the next section. Since a smaller $A$ means that we restrict the \blue{SG} of $\phi$ in~\eqref{eq:lure} to a smaller set of signals, the gains $\lambda_{\omega, A} ,\lambda^\partial_{\omega, A}$ are monotone increasing functions in $A$, where $\lambda^{(\partial)}_{\omega, A} \xrightarrow[]{A \to \infty} \lambda^{(\partial)}_{\omega} $ and $\lambda_{\omega, A} \leq \lambda_{\omega}$, $\lambda^\partial_{\omega, A} \leq \lambda^\partial_{\omega}$. Using the gains $\lambda_\omega, \lambda^\partial_\omega$,~\eqref{eq:sobolev} bounds the amplitude of the output as
\begin{equation}\label{eq:amplitude-bound-A-indep}
\begin{aligned}
    \maketextstyle &\norm{\tilde{y}}_\infty \leq \sqrt{2 \norm{\tilde{y}}_T \norm{\tilde{y}'}_T} \hfill  \\ &\leq \sqrt{2 \lambda_{\omega} \lambda^\partial_{\omega} \norm{u}_T \norm{u'}_T} \leq  \sqrt{2 \lambda_{\omega} \lambda^\partial_{\omega} U}.
\end{aligned}
\end{equation}
From~\eqref{eq:amplitude-bound-A-indep}, we obtain the bound $\norm{\tilde{y}}_\infty \leq A = \sqrt{2 \lambda_{\omega} \lambda^\partial_{\omega} U}$. We use this bound $\norm{\tilde{y}}_{\mathblue{\infty}} \leq A$ to compute the amplitude-dependent gains
$\lambda_{\omega, A} $, $\lambda^\partial_{\omega, A} $, which yields a tighter bound
\begin{equation}\label{eq:amplitude-bound}
    \maketextstyle \norm{\tilde{y}}_\infty  \leq  \sqrt{2 \lambda_{\omega, A} \lambda^\partial_{\omega, A} U} \leq A.
\end{equation}
To find the tightest bound $\norm{\tilde{y}}_{\mathblue{\infty}} \leq A$, one must find the \emph{smallest} $A$ that satisfies~\eqref{eq:amplitude-bound}. As we are interested in how the amplitude varies as a function of the input frequency, we solve the following optimization problem for each $\omega \in \R_{\geq 0}$
\begin{equation}\label{eq:amplitude-bound-optimization}
    \maketextstyle A_{\omega,U} = \inf  \left\{ A \geq 0 \mid ~\eqref{eq:amplitude-bound} \text{ holds }  \right\},
\end{equation}
which is guaranteed to have a solution if $\lambda_{\omega} \lambda^\partial_{\omega} < \infty$ in~\eqref{eq:amplitude-bound-A-indep}. For a fixed $\omega$ and $U$, the minimizer in~\eqref{eq:amplitude-bound-optimization} is found using bisection with initial lower bound $A_\mathrm{min}=0$ and upper bound $A_\mathrm{max} = \sqrt{2 \lambda_{\omega} \lambda^\partial_{\omega} U}$ from~\eqref{eq:amplitude-bound-A-indep}. 

Given that the Lur'e system is \PP{}, we will show in the next section how SRG methods can be used to: 1) guarantee that $u \in \mathcal{U}_\omega \cap H^1[0,T]$ implies $\tilde{e}, \tilde{y} \in \mathcal{U}_\omega \cap H^1[0,T]$ and 2) compute the gains $\lambda_{\omega,A} $ and $\lambda^\partial_{\omega,A}$.

\section{Gain Computation using SRG Methods}\label{sec:gain-from-SRGs}

In this section, we explain how the SG can be computed for LTI and slope-restricted static NL functions, and derive the SG for their derivative maps. Then, we show how they are used to transfer $H^1$ properties from input to output, and to compute the gains in~\eqref{eq:gamma_w} and~\eqref{eq:gamma_w_U}.

\subsection{\blue{SG} of LTI Systems}

Throughout this section, we assume that $G:L_2 \to L_2$ is a proper and stable LTI system with transfer function $G(s)$.

\subsubsection{Input to output}

From~\cite{chaffeyGraphicalNonlinearSystem2023,patesScaledRelativeGraph2021} we know that $\SG(G)$ is the \emph{hyperbolic convex hull} $\hco(\cdot)$ (see~\cite{patesScaledRelativeGraph2021} for the definition) of the Nyquist diagram
\begin{equation}\label{eq:srg_G}
    \SG(G) = \hco(\{ G(j \omega) \mid \omega \in \R \}).
\end{equation}
For an input $u \in \mathblue{\mathcal{U}_{\omega \mathrm{e}}}$, the output $y=Gu \mathblue{ \in \Lte}$ converges exponentially to an odd $T$-periodic output $\tilde{y}(t) = \sum_{k \in 2\Z+1} G(j k \omega) \hat{u}_k e^{j k \omega t}$~\cite[Ch. 2]{skogestadMultivariableFeedbackControl2010},
where $\tilde{y} \in \mathblue{\mathcal{U}_{\omega \mathrm{e}}}$ since $G$ is stable~\cite{krebbekxNonlinearBandwidthBode2025a}. \blue{As we can view $\tilde{y} \in \mathcal{U}_\omega $, and} $\mathcal{U}_\omega$ is a linear subspace of $L_2[0,T]$, it follows from~\cite{patesScaledRelativeGraph2021} that the \blue{SG} of the map $\mathblue{\tilde{G}} : \mathcal{U}_\omega \to \mathcal{U}_\omega$, \blue{defined as $\tilde{G}: u \mapsto \tilde{y}$}, is given by\footnote{If $u$ were periodic, but not in $\mathblue{\mathcal{U}_{\omega \mathrm{e}}}$, it can have a nonzero DC component, which would force us to include $G(0)$ in the hyperbolic convex hull. This can lead to a larger \blue{SG}, i.e., more conservatism. To avoid this, the odd half-wave symmetry in $\mathblue{\mathcal{U}_{\omega \mathrm{e}}}$ is a \emph{sufficient} condition for a zero DC component.}%
\begin{equation}\label{eq:srg_G_w}
    \SG_{\mathcal{U}_\omega}(\mathblue{\tilde{G}}) = \hco(\{ G(j k \omega) \mid k \in 2\Z+1 \}),
\end{equation}
as was proven in~\cite{krebbekxNonlinearBandwidthBode2025a}. %

\subsubsection{Derivative map}\label{sec:lti_derivative_map} 

The following two propositions show how the \blue{SG} of the derivative map $\partial G$ is related to $\SG(G)$.

\begin{proposition}\label{prop:srg_G_deriv}
    For a stable and proper $G(s)$, the \blue{SG} of $\partial G : u' \mapsto y'$ restricted to $u \in H^1_0$ is given by 
    \begin{equation}\label{eq:srg_G_deriv}
        \maketextstyle \SG_{\mathcal{D}_0}(\partial G) \subseteq \SG(G),
    \end{equation}
    where $\mathcal{D}_0 = \{ u' \in L_2 \mid u \in H_0^1 \}$ is the domain of $\partial G$.
\end{proposition}

\begin{proof}
    Let $u \in  H^1_0$ and $\partial G$ denote the map $\partial G : u' \mapsto y'$, where $y= Gu$. The Laplace transform of a (weak) derivative reads $\mathcal{L}(x') = sX(s) - x(0)$, and since $u \in H^1_0$ and we assume zero initial conditions for $G$, the operator $\partial G$ in the Laplace domain reads $sY(s) = G(s) s U(s)$. Moreover, since $u \in H_0^1 \iff u' \in \mathcal{D}_0$, we can conclude that $\partial G : \mathcal{D}_0 \to L_2$. Finally, the maps $\partial G : \mathcal{D}_0 \to L_2$ and $G : L_2 \to L_2$ are defined by the same Laplace domain operator $G(s)$, and since $\mathcal{D}_0 \subseteq L_2$, we conclude $\SG_{\mathcal{D}_0}(\partial G) \subseteq \SG(G)$. 
\end{proof}
\vspace*{-1em}
\begin{proposition}\label{prop:srg_w_g_deriv}
    For a stable and proper $G(s)$, the \blue{SG} of $\partial \mathblue{ \tilde{G} }: u' \mapsto \tilde{y}'$ restricted to $u \in \mathcal{U}_\omega \cap H^1[0,T]$ is given by
    \begin{equation}\label{eq:srg_w_G_deriv}
        \maketextstyle \SG_{\mathcal{U}_\omega}(\partial  \mathblue{ \tilde{G} }) = \SG_{\mathcal{U}_\omega}( \mathblue{ \tilde{G} }).
    \end{equation}
\end{proposition}
\begin{proof}
    Let $u \in \mathcal{U}_\omega \cap H^1[0,1]$ and $\partial \mathblue{ \tilde{G} }$ denote the map $\partial \mathblue{ \tilde{G} } : u' \mapsto \tilde{y}'$, where $\tilde{y}$ is the odd $T$-periodic response to $u$. The Poincar\'e inequality $u' \in L_2[0,T] \implies u \in L_2[0,T]$~\cite[p. 218]{brezisFunctionalAnalysisSobolev2011} and stability of $G(s)$ imply $\partial \mathblue{ \tilde{G} }: \mathcal{U}_\omega \to \mathcal{U}_\omega$. Since $\tilde{y}'(t) = \sum_{k \in 2\Z+1} G(j k \omega) \hat{u}'_k e^{j k \omega t}$,
    where $\hat{u}'_k = jk\omega \hat{u}_k$, it is clear that $\SG_{\mathcal{U}_\omega}(\partial \mathblue{ \tilde{G} }) = \SG_{\mathcal{U}_\omega}(\mathblue{ \tilde{G} })$ by the same reasoning as in Proposition~\ref{prop:srg_G_deriv}, where we obtain equality since $\mathblue{ \tilde{G} }$ and $\partial \mathblue{ \tilde{G} }$ both map $\mathcal{U}_\omega$ \blue{in}to itself.
\end{proof}

The procedure above can be repeated to obtain the \blue{SG} for higher order derivative maps, which is not pursued here.

\subsection{\blue{SG} of Nonlinearities}\label{sec:nonlinearities}

Throughout this section, consider the NL function $\phi : \R \to \R$ which satisfies for all $A>0$ and $|x|, |y| \leq A$
\begin{equation}\label{eq:static_nl}
    \maketextstyle a(A) \leq \frac{\phi(x) - \phi(y)}{x-y} \leq b(A), \quad c(A) \leq \frac{\phi(x)}{x}\leq d(A),
\end{equation}
where $a,b,c,d : \R_{\geq 0} \mapsto \R$ are bounded monotone functions $a(A) \leq b(A), c(A) \leq d(A)$. Their limits as $A \to \infty$ are denoted as $a^*, b^*, c^*, d^*$. We use $\phi_A \in \partial [a(A),b(A)]$ ($\phi \in \partial [a^*, b^*]$) and $ \phi_A \in [c(A),d(A)]$ ($\phi \in [c^*, d^*]$) to denote the amplitude-(in)dependent slope and sector condition, respectively\footnote{Note that $[c(A), d(A)] \subseteq [a(A), b(A)]$, and $[a(A_1), b(A_1)] \subseteq [a(A_2), b(A_2)]$ and $[c(A_1), d(A_1)] \subseteq [c(A_2), d(A_2)]$ for $0\leq A_1 \leq A_2$.}.

\subsubsection{Input to output}

From~\cite{chaffeyGraphicalNonlinearSystem2023}, we know that for $\phi$ in~\eqref{eq:static_nl}, the SG obeys $\SG(\phi) \subseteq D_{[c^*, d^*]}$, hence $\gamma(\phi) \leq \max\{|c^*|, |d^*|\}$. As shown in~\cite{chaffeyGraphicalNonlinearSystem2023}, this bound is exact in specific cases, but in general it is unknown how conservative the bound is. Moreover, since $\phi$ is static, it follows that $\SG_{\mathcal{U}_\omega}(\phi) = \SG_{\mathcal{U}_{\omega'}}(\phi)$ for any $\omega,\omega' \in \R_{\geq 0}$ and $\SG_{\mathcal{U}_\omega}(\phi) \subseteq \SG(\phi)$. Similarly, for the amplitude-dependent case, we have %
\begin{equation*}\label{eq:A_dep_SG}
    \SG(\phi_A) \subseteq D_{[c(A), d(A)]},
\end{equation*}
where $\SG(\phi_A)$ denotes the $\SG$ of $\phi$ restricted to $y \in L_2$ with $\norm{y}_\infty \leq A$, where $A$ must satisfy~\eqref{eq:amplitude-bound}.

\subsubsection{Derivative map}

Now consider some $y \in H^1$ and let $v = \phi(y)$. Similarly to~\ref{sec:lti_derivative_map},  we wish to prove $v' \in L_2$ and analyze the map $\partial \phi : L_2 \to L_2$ defined by $\partial \phi : y' \mapsto v'$.

If the input obeys $\norm{y}_\infty \leq A$, then the $L_2$-gain of this operator reads as follows
\vspace{-.8em}\begin{gather}
    \textstyle \gamma(\partial \phi_A) := \sup_{\substack{y,y' \in L_2 \\ \norm{y}_\infty \leq A}} \frac{\sqrt{\int_0^\infty | \phi(y(t))'|^2 \dd t}}{\sqrt{\int_0^\infty |y'(t)|^2 \dd t}} \nonumber \\ \textstyle = \sup_{\substack{y,y' \in L_2 \\ \norm{y}_\infty \leq A}} \frac{\sqrt{\int_0^\infty |\phi'(y(t)) y'(t)|^2 \dd t}}{\sqrt{\int_0^\infty |y'(t)|^2 \dd t}} \label{eq:A_dep_SRG_deriv} \\   \textstyle \leq \esssup_{\substack{y,y' \in L_2  \\ \norm{y}_\infty \leq A}} | \phi'(y) | \frac{\sqrt{\int_0^\infty | y'(t)|^2 \dd t}}{\sqrt{\int_0^\infty |y'(t)|^2 \dd t}} \displaystyle = \esssup_{|s| \leq A} | \phi'(s) |. \nonumber
\end{gather}
Since $\phi$ satisfies a Lipschitz bound for all $A$, it obeys the chain rule $v'(t) = \phi(y(t))' = \phi'(y(t)) y'(t)$ a.e., and the weak derivative obeys $|\phi'(s)| \leq \max\{|a(A)|, |b(A)|\}$ for all $s\in \R$ a.e., see~\cite[p. 48]{ziemerWeaklyDifferentiableFunctions1989}. 

\begin{proposition}\label{prop:A_dep_deriv_SG_bound}
    Let $\phi$ be given as in~\eqref{eq:static_nl} and $A>0$, then
    \begin{equation}\label{eq:A_dep_deriv_SG_bound}
        \SG(\partial \phi_A) \subseteq D_{[a(A), b(A)]},
    \end{equation}
    where $\SG(\partial \phi_A)$ denotes the SG of $\partial \phi$ restricted to inputs $y'$ that have $y \in H^1$ and $\norm{y}_\infty \leq A$. 
\end{proposition}
\begin{proof}
    Fix some $A>0$ and take $k \in \R$ such that $-a(A)-k = b(A)+k =: \kappa$, then the shifted nonlinearity $\varphi_A := \phi_A + \mathblue{k}$ satisfies $\varphi_A \in \partial[-\kappa,\kappa]$, which by~\eqref{eq:A_dep_SRG_deriv} implies $\gamma(\partial \varphi_A) \leq \kappa$. Then, from~\cite{ryuScaledRelativeGraphs2022} we can conclude that $\SG(\partial \varphi_A) \subseteq D_{[-\kappa, \kappa]}$.
    
    Since $\partial \varphi_A (y') =  [\phi_A(y) + ky]' = \partial \phi_A (y') + ky'$, it follows that $\partial \phi_A = \partial \varphi_A - k$. Using SRG identities from~\cite{ryuScaledRelativeGraphs2022}
    \begin{equation*}
        \SG(\partial \phi_A) = \SG(\partial \varphi_A) - k \subseteq D_{[-\kappa-k, \kappa-k]} = D_{[a(A), b(A)]},
    \end{equation*}
    which proves~\eqref{eq:A_dep_deriv_SG_bound}.
\end{proof}

\subsection{Analysis of the Lur'e System}

Throughout \blue{sections~\ref{subsubsec:lure_stab} and~\ref{subsubsec:gains}}, we assume that Condition~\ref{condition:lure} holds and that~\eqref{eq:lure} is \PP{}.

\subsubsection{Stability of the Lur'e system}\label{subsubsec:lure_stab}

First, $L_2$-stability of~\eqref{eq:lure} is established using Proposition~\ref{prop:lure_stability_theorem}. Since $\SG(\phi) \subseteq D_{[c^*, d^*]}$ is a disk, it automatically satisfies the chord property, and if the conditions of Proposition~\ref{prop:lure_stability_theorem} are satisfied, then $\dist(\SG(G)^{-1}, -D_{[c^*, d^*]}) = r$ implies $\gamma([G, \phi]) \leq \frac{1}{r}$. 

Second, we consider $u \in H^1_0$. Since $\SG(\partial \phi) \subseteq D_{[a^*, b^*]}$ by Proposition~\ref{prop:A_dep_deriv_SG_bound} and $\SG_{\mathcal{D}_0}(\partial G) \subseteq \SG(G)$ \blue{by Proposition~\ref{prop:srg_G_deriv}}, if $\partial G$ and $\partial \phi$ satisfy Proposition~\ref{prop:lure_stability_theorem}, then $\dist(\SG(G)^{-1}, - D_{[a^*, b^*]}) = r^\partial$
implies $\norm{y'} \leq \frac{1}{r^\partial} \norm{u'}$, hence $y \in H^1_0$. We see that SRG analysis for $u \mapsto y$ and its derivative can \blue{imply} smoothness \blue{of the output $y$} from the input $u$.

\subsubsection{Computing the gains $\lambda_{\omega, A}$ and $\lambda^\partial_{\omega, A}$}\label{subsubsec:gains}

Once stability of the Lur'e system is established, one can study the gain as a function of frequency and amplitude. The following lemma shows how $\lambda_{\omega, A}$ and $\lambda^\partial_{\omega, A}$ in~\eqref{eq:amplitude-bound} are computed. 

\begin{lemma}\label{lemma:lambdas}
    If Condition~\ref{condition:lure} holds, $\gamma([G, \phi_A])<\infty$, $\norm{y}_\infty \leq A$, $u \in \mathcal{U}_\omega$ and~\eqref{eq:lure} is \PP{}, then
    \begin{equation}\label{eq:r_w_A}
        \dist(\SG_{\mathcal{U}_\omega}(\mathblue{ \tilde{G} })^{-1}, -D_{[c(A), d(A)]}) = r_{\omega, A}>0
    \end{equation}
    implies $\norm{\tilde{y}}_T \leq \frac{1}{r_{\omega, A}} \norm{u}_T$, hence $\lambda_{\omega, A} \leq \frac{1}{r_{\omega, A}}$. If, in addition, $\gamma([G, \partial \phi_A])<\infty$ and $u \in \mathcal{U}_\omega \cap H^1[0,T]$, then
    \begin{equation}\label{eq:r_partial_w_A}
        \dist(\SG_{\mathcal{U}_\omega}(\mathblue{ \tilde{G} })^{-1}, -D_{[a(A), b(A)]}) = r^\partial_{\omega, A}>0
    \end{equation}
    implies $\norm{\tilde{y}'}_T \leq \frac{1}{r^\partial_{\omega, A}} \norm{u'}_T$, hence $\lambda^\partial_{\omega, A} \leq \frac{1}{r^\partial_{\omega, A}}$.
\end{lemma}

\begin{proof}
     Let $u \in \mathcal{U}_\omega$. Since $\gamma([G, \phi_A])<\infty$ and~\eqref{eq:lure} is \PP{}, we know $\tilde{e}, \tilde{y} \in \mathcal{U}_\omega$, hence the \blue{SGs} of $G$ and $\phi_A$ are given by~\eqref{eq:srg_G_w} and $D_{[c(A), d(A)]}$, respectively. Note that we only need to check~\eqref{eq:lure_stability_theorem} at $\tau=1$ since $\gamma([G, \phi_A])<\infty$. Therefore,~\eqref{eq:r_w_A} with Proposition~\ref{prop:lure_stability_theorem} implies $\norm{\tilde{y}}_T \leq \frac{1}{r_{\omega, A}} \norm{u}_T$.

     Under the additional conditions, we can conclude $\tilde{e}, \tilde{y} \in \mathcal{U}_\omega \cap H^1[0,T]$, hence we can use Propositions~\ref{prop:srg_w_g_deriv} and~\ref{prop:A_dep_deriv_SG_bound} for $\partial \mathblue{ \tilde{G} }$ and $\partial \phi_A$, respectively. Again,~\eqref{eq:r_partial_w_A} with Proposition~\ref{prop:lure_stability_theorem} implies $\norm{\tilde{y}'}_T \leq \frac{1}{r^\partial_{\omega, A}} \norm{u'}_T$.
\end{proof}

The $A \to \infty$ limit is obtained from Lemma~\ref{lemma:lambdas} by taking $a(A) \to a^*, b(A) \to b^*, c(A) \to c^*, d(A) \to d^*$. The resulting margins are denoted $r_\omega, r^\partial_\omega$ and are used to bound $\lambda_\omega \leq \frac{1}{r_\omega}$ and $\lambda^\partial_\omega \leq \frac{1}{r^\partial_\omega}$ in~\eqref{eq:amplitude-bound-A-indep}. 

Only amplitude-dependent gain bounds for $u \in H^1_0$ can be derived considering the $\omega \to 0$ limit in Lemma~\ref{lemma:lambdas}. In that case, Proposition~\ref{prop:srg_G_deriv} is used instead of~\ref{prop:srg_w_g_deriv} in the proof.

\subsubsection{Main result for Lur'e systems}

We can now combine the results from this section with~\eqref{eq:amplitude-bound-optimization} to compute the amplitude and frequency dependent gain in~\eqref{eq:gamma_w_U}.

\begin{theorem}\label{thm:lure}
    Consider~\eqref{eq:lure} where $G$ is stable, rational and strictly proper, $\phi$ is odd and obeys~\eqref{eq:static_nl}, and
    \begin{equation}\label{eq:thm-lure}
        \dist(\SG(G)^{-1}, - \tau D_{[a^*, b^*]}) = r > 0, \quad \forall \tau \in [0, 1].
    \end{equation}
    Then for $R = [G,\phi]$ and all $\omega \in \R_{\geq 0}$ and $\norm{u}_\infty < \infty$,
    \begin{enumerate}
        \item[1)] $R$ is well-posed, and \PP{}  for all $u \in \mathblue{\mathcal{U}_{\omega \mathrm{e}}}$,
        \item[2)] $\gamma_\omega(R) \leq \frac{1}{r_{\omega}} < \infty$,
        \item[3)] for all $\omega, U \in \R_{\geq 0}$ such that $\norm{u}_T \norm{u'}_T \leq U$, \eqref{eq:amplitude-bound-optimization} has a solution $A_{\omega,U} < \infty$ and $\gamma_{\omega, U}(R) \leq \frac{1}{r_{\omega, A_{\omega,U}}}$.
    \end{enumerate}
\end{theorem}

\begin{proof}

    First, we show that $[G, \tau \phi]$ is well-posed for all $\tau \in [0,1]$. From~\cite{chaffeyGraphicalNonlinearSystem2023} it follows that $\SRG(\phi) \subseteq D_{[a^*, b^*]}$, hence~\eqref{eq:thm-lure} implies that $[G, \tau \phi]$ is \emph{incrementally} stable and well-posed using an incremental homotopy argument~\cite{chaffeyHomotopyTheoremIncremental2025}.

    Then, we prove that $R$ is \PP{} for $u \in \mathblue{\mathcal{U}_{\omega \mathrm{e}}}$. Note that~\eqref{eq:thm-lure} implies that $\mathblue{G_{a}}(s):= \frac{G(s)}{1+a^* G(s)}$ is stable since $a^* \in D_{[a^*, b^*]}$. Since $r>0$, the inequality~\eqref{eq:thm-lure} implies $-D_{[0, \mu+r/2]} \cap \SG(\mathblue{G_{a}})^{-1} = \emptyset$, which is equivalent to
    \begin{equation}\label{eq:separation-proof-eq}
        \maketextstyle (-D_{[0, \mu + r/2]})^{-1} \cap \SG(\mathblue{G_{a}}) = \emptyset,
    \end{equation}
    where $\SG(\mathblue{G_{a}})^{-1} = a^*+ \SG(G)^{-1}$ and $\mu = b^*-a^*$ and $(-D_{[0, \mu + r/2]})^{-1} = \C_{\mathrm{Re} \leq -\frac{1}{\mu+r/2}}$. Since the \blue{SG} contains the Nyquist diagram~\cite{chaffeyGraphicalNonlinearSystem2023} and $\frac{1}{\mu+r/2} < \frac{1}{\mu}$,~\eqref{eq:separation-proof-eq} implies $\mathrm{Re}(\mathsf{C} (j\omega I - \mathsf{A})^{-1} \mathsf{B}) > -\frac{1}{\mu}$ for all $\omega \in \R$, where $\dot x = \mathsf{A}x+\mathsf{B} u, \, y = \mathsf{C} x$ is a state-space representation\footnote{As $G$ is strictly proper, so is $\mathblue{G_{a}}$, and there is no feedthrough term $\mathsf{D}u$.} for $y = \mathblue{G_{a}}u$. By the conditions in~\cite{pavlovFrequencyDomainPerformance2007}, which require oddness $\phi(-x) = -\phi(x)$, the Lur'e system with $\mathblue{G_{a}}$ and $\tilde{\phi} := \phi - a^*$ is exponentially convergent for all bounded piecewise continuous inputs. This means that for any $u \in \mathblue{\mathcal{U}_{\omega \mathrm{e}}}$ with $\norm{u}_\infty < \infty$, there exist unique steady-state solutions $\tilde{e}(t+T) = \tilde{e}(t), \tilde{y}(t+T)= \tilde{y}(t)$. Oddness of $\phi$ implies that the state-space representation $\dot x = f(x, u) = \mathsf{A}x + \mathsf{B}(u-\phi(\mathsf{C}x)), \, y = \mathsf{C} x$ has the symmetry $f(-x,-u) = -f(x, u)$. This means that the input $u_2(t) =u(t+T/2) $ yields steady-state outputs $\tilde{e}_2 = -\tilde{e}, \tilde{y}_2 = -\tilde{y}$, which implies odd $T$-periodicity of $\tilde{e}, \tilde{y}$. Since $D_{[c^*,d^*]} \subseteq D_{[a^*, b^*]}$,~\eqref{eq:thm-lure} implies $\gamma(R) \leq \frac{1}{r}$ and hence $\tilde{e}, \tilde{y} \in \mathcal{U}_\omega$ holds, proving part 1). 

    Part 2) follows from the fact that $\SG_{\mathcal{U}_\omega}(\mathblue{\tilde{G}}) \subseteq \SG(G)$ and $D_{[c^*,d^*]} \subseteq D_{[a^*, b^*]}$. The bound $\gamma_\omega(R) \leq \frac{1}{r_\omega}$ follows from Lemma~\ref{lemma:lambdas} with $A \to \infty$. 
    
    For part 3), note that $U< \infty$ implies $\norm{u}_\infty < \infty$ by~\eqref{eq:sobolev}, therefore the system is \PP{} for all $u \in \mathblue{\mathcal{U}_{\omega \mathrm{e}}}$ and $\norm{u}_T \norm{u'}_T \leq U$. Since~\eqref{eq:thm-lure} guarantees that $r_\omega, r^\partial_\omega < \infty$, the solution~\eqref{eq:amplitude-bound-optimization} satisfies $A_{\omega,U} \leq  \sqrt{2(r_\omega r^\partial_\omega)^{-1} U}$, hence $\gamma_{\omega, U}(R) \leq \frac{1}{r_{\omega, A_{\omega,U}}}$ by Lemma~\ref{lemma:lambdas} with $A = A_{\omega,U}$. 
\end{proof}

Although they are not pursued here, there are some immediate implications of the method in this paper. First, it is possible to ``warm start''~\eqref{eq:amplitude-bound-optimization} with some a priori bound $\norm{y}_\infty \leq A$. This opens up the possibility of analyzing the performance of systems that are locally stable, but unstable for large amplitudes, e.g., systems with cubic nonlinearities. Second, one can use Lemma~\ref{lemma:lambdas} with a different $L_1$-method to compute $\norm{y}_\infty \leq A$ as in~\cite{megretskiCombiningL1L21995}, either replacing~\eqref{eq:amplitude-bound-A-indep} or in combination. Finally, Theorem~\ref{thm:lure} yields more system properties; $H^1$-smoothness of the output, finite derivative gain $u' \mapsto y'$, and an amplitude-dependent gain bound $\gamma_U$.

\section{Example}\label{sec:example}

Consider the Lur'e system in Fig.~\ref{fig:lure} where
\begin{equation}\label{eq:example_lure}
    \maketextstyle G(s) = \frac{1}{s+2}, \quad \phi(x) = \sin(x).
\end{equation}
Systems like~\eqref{eq:example_lure}, composed of an LTI filter and a sinusoidal nonlinearity, are used as simple models for PLLs, see~\cite{stephensPhaseLockedLoopsWireless1998}. 

For details on how to compute $\SG_{\mathcal{U}_\omega} (G)$, we refer the reader to~\cite{krebbekxNonlinearBandwidthBode2025a}. For $\phi(x) = \sin(x)$, the functions $a,b,c,d$ in~\eqref{eq:static_nl} are given by $b(A) = d(A) = 1$ and 
\begin{equation*}
    a(A) = 
    \begin{cases}
        \cos(A) &\hspace*{-0.5em}\text{if } A \leq \pi,\\
        -1 &\hspace*{-0.5em}\text{else},
    \end{cases}
    \, c(A) = 
    \begin{cases}
        \frac{\sin(A)}{A} &\hspace*{-0.5em}\text{if } A \leq A^*,\\
        \frac{\sin(A^*)}{A^*} &\hspace*{-0.5em}\text{else},
    \end{cases}
\end{equation*}
where $A^* \approx 4.4934$ is the global minimizer of $\frac{\sin(A)}{A}$ for $A \in \R_{\geq 0}$. Since all conditions for Theorem~\ref{thm:lure} are met, we can compute the gains in~\eqref{eq:gamma_w} and~\eqref{eq:gamma_w_U}. 

\begin{figure}[t]
    \centering
    \begin{subfigure}[b]{\linewidth}
        \centering
        \includegraphics[width = \linewidth]{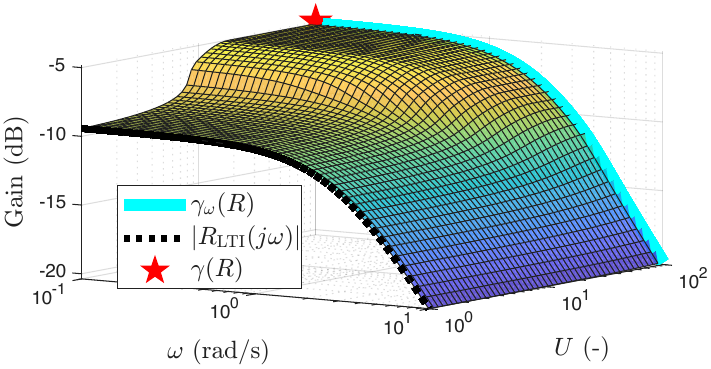}
        \caption{$L_2$-gain $\norm{\tilde{y}}_T \leq \gamma_{\omega, U} \norm{u}_T$ in~\eqref{eq:gamma_w_U}.}
        \label{fig:surf_gamma}
    \end{subfigure}
    \hfill
    \begin{subfigure}[b]{\linewidth}
        \centering
        \includegraphics[width = \linewidth]{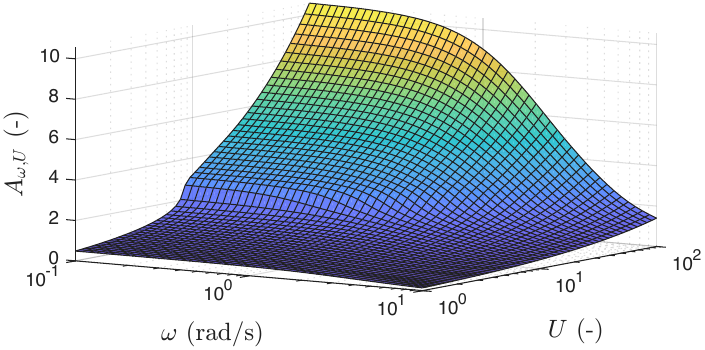}
        \caption{Amplitude $\norm{\tilde{y}}_\infty \leq A_{\omega, U}$ in~\eqref{eq:amplitude-bound-optimization} as a function of $\omega$, $U$.}
        \label{fig:surf_A}
    \end{subfigure}
    \caption{Gain and amplitude \blue{bound} for example~\eqref{eq:example_lure}.}
    \label{fig:surfs}
\end{figure}

The result is given in Fig.~\ref{fig:surf_gamma}, where it is clear that $\gamma_\omega(R) = \lim_{U \to \infty} \gamma_{\omega, U}(R)$, for each $\omega$. At the $U \to 0$ limit, we see that the ``LTI'' behavior is recovered, which is described by $R_\mathrm{LTI}(s) = \frac{G(s)}{1+ G(s)}$, obtained by replacing $\sin(x) \to x$, i.e., its linearization at $x=0$. Between these limit cases, $\gamma_{\omega, U}$ interpolates between the worst-case NL gain $\gamma_\omega$, and the LTI performance $|R_\mathrm{LTI}(j \omega)|$. Finally, at $\omega \to 0, U\to \infty$, the standard $L_2$-gain $\gamma(R)$ in~\eqref{eq:L2-gain} is attained.

The gains $\gamma_{\omega, U}$ are computed using the solution~\eqref{eq:amplitude-bound-optimization}, which bound $\norm{y}_\infty \leq A_{\omega, U}$ for that $\omega, U$. These are plotted in Fig.~\ref{fig:surf_A}. This figure can be used to determine which constraint $U$ must be used to certify a certain amplitude bound on $y$. In addition, we have proven that for all inputs in $H^1_0$, the resulting output is in $H^1_0$.

\section{Conclusion}\label{sec:conclusion}

This paper develops a graphical frequency- and amplitude-dependent Bode diagram for analyzing the performance of NL feedback systems, generalizing the Bode plot for LTI systems. We focus on Lur'e systems and bound the amplitude of the output in terms of the $L_2$-norm of the output and its derivative using Sobolev theory. Then, SRG methods are used to compute the \blue{$L_2$-}gain from input (derivative) to output (derivative). We show that our frequency- and amplitude-dependent performance metric interpolates between the worst-case $L_2$-gain of the NL system, and the linearized LTI model. We derive a practical theorem for Lur'e systems with slope-restricted nonlinearities and finally demonstrate the effectiveness of our method on a Lur'e system that represents PLL dynamics.

Future research includes multivariable systems, general interconnections as in~\cite{krebbekxNonlinearBandwidthBode2025a}, broader classes of nonlinearities, including even harmonics and DC gain, and loop-shaping.

\footnotesize
\bibliographystyle{IEEEtran} 
\bibliography{bibliography} 

\end{document}